\documentclass[11pt,a4paper]{article}
\usepackage{imports}
\newcommand{\argmin}{\operatorname*{argmin}}
\newcommand{\argmax}{\operatorname*{argmax}}

\newcommand{\polylog}{\operatorname{polylog}}
\newcommand{\tO}{\ensuremath{\widetilde O}}
\newcommand{\hO}{\ensuremath{\widehat O}}

\newcommand{\abs}[1]{\ensuremath{\left| #1 \right|}}

\def\Bigbar#1{\mathrel{\left|\vphantom{#1}\right.\n@space}}

\newcommand{\eps}{\epsilon}
\newcommand{\R}{\mathbb{R}}

\let\oldvec\vec \renewcommand{\vec}[1]{\oldvec{\mkern0mu#1}}

\newcommand{\samplesize}{\ensuremath{5390 \cdot n \log^2(n)/\eps}}
\newcommand{\Osamplesize}{\ensuremath{O(n \log^2 (n) / \eps)}}
\newcommand{\fn}[1]{\textup{\textsc{#1}}}

\newcommand{\bound}{\abs{w(\vec{C} \cap F_i) - u(\vec{C} \cap F_i)} \le \frac{\eps}{2} \cdot \left( \frac{u(C \cap E_i) \cdot 2^{i-1}}{\pi_i \cdot \gamma \cdot (\beta+1)} + u(\vec{C} \cap F_i) \right)}

\makeatletter
\DeclareRobustCommand{\cev}[1]{%
  {\mathpalette\do@cev{#1}}%
}
\newcommand{\do@cev}[2]{%
  \vbox{\offinterlineskip
    \sbox\z@{$\m@th#1 x$}%
    \ialign{##\cr
      \hidewidth\reflectbox{$\m@th#1\vec{}\mkern4mu$}\hidewidth\cr
      \noalign{\kern-\ht\z@}
      $\m@th#1#2$\cr
    }%
  }%
}
\makeatother

\newcommand{\generalsamplesize}{\ensuremath{\frac{128 \gamma (\beta+1) \ln n}{0.38 \eps^2} \cdot \sum_{e \in E} \lambda_e^{-1}}}

\NewDocumentCommand{\fancydef}{O{} O{} m m}{%
    \expandafter\def\csname #3\endcsname{\hyperref[def:#3]{\textcolor{black}{#4}}\xspace}%
    \IfNoValueTF{#1}{%
        \expandafter\def\csname #3s\endcsname{\hyperref[def:#3]{\textcolor{black}{#4s}}\xspace}%
    }{%
        \expandafter\def\csname #1\endcsname{\hyperref[def:#3]{\textcolor{black}{#2}}\xspace}%
    }%
}

\usepackage[normalem]{ulem}

\begin{document}
\title{Incremental Approximate Maximum Flow \\ via Residual Graph Sparsification}
\date{}
\author[1]{Gramoz Goranci}
\affil[1,4]{Faculty of Computer Science, University of Vienna, Austria}
\author[2]{Monika Henzinger}
\affil[2]{Institute of Science and Technology Austria (ISTA), Klosterneuburg, Austria}
\author[3]{Harald R\"acke}
\affil[3]{Technical University Munich, Germany}
\author[4]{A. R. Sricharan}
\affil[4]{UniVie Doctoral School Computer Science DoCS}

\maketitle

\begin{abstract}
We give an algorithm that, with high probability, maintains a $(1-\epsilon)$-approximate $s$-$t$ maximum flow in undirected, uncapacitated $n$-vertex graphs undergoing $m$ edge insertions in $\tilde{O}(m+ n F^*/\epsilon)$ total update time, where $F^{*}$ is the maximum flow on the final graph. This is the first algorithm to achieve polylogarithmic amortized update time for dense graphs ($m = \Omega(n^2)$), and more generally, for graphs where $F^*= \tilde{O}(m/n)$.

At the heart of our incremental algorithm is the residual graph sparsification technique of Karger and Levine [STOC '02, SICOMP '15], originally designed for computing exact maximum flows in the static setting. Our main contributions are (i) showing how to maintain such sparsifiers for approximate maximum flows in the incremental setting and (ii) generalizing the cut sparsification framework of Fung et al. [STOC '11, SICOMP '19] from undirected graphs to balanced directed graphs.
 \end{abstract}

\section{Introduction}
\label{sec:intro}

The maximum $s$-$t$ flow problem has been at the forefront of research in theoretical computer science and combinatorial optimization. Algorithms developed for maximum $s$-$t$ flow and its dual problem, minimum $s$-$t$ cut, have been highly influential due to their wide applicability~\cite{ahuja1995applications} and their use as subroutines in other algorithms~\cite{arora2012multiplicative, sherman2009breaking}.
A long line of work has improved our understanding of how efficiently maximum flow can be solved in the static setting. Two landmark combinatorial results are the deterministic algorithm by Goldberg and Rao~\cite{GR98flow}, which achieves a running time of $O(\min \{ n^{2/3}, m^{1/2}  \} \cdot m)$, and the randomized algorithm by Karger and Levine~\cite{KL15sampling} which has a running time of $\tO(m + nF^*)$,\footnote{In this paper, we use $\tO(X)$ to denote $O(X \polylog(X))$ and $\hO(X)$ to denote $O(X^{1+o(1)})$.} where $F^*$ is the value of the maximum flow.
More recently, efforts building upon a novel blend of continuous optimization techniques, the Laplacian paradigm~\cite{teng2010laplacian} and graph-based data structures have led to even faster algorithms; notably, a randomized $\tO(m/\eps)$ time approximate maximum flow algorithm on undirected graphs~\cite{She13almostflow, KLOS14almost, Pen16nearly, She17areaconvex}, and a deterministic $\hO(m)$ time \emph{exact} algorithm for min-cost flow (and thus maximum flow) even on directed graphs~\cite{CKL22almost, CKL24incremental}, which constitutes a major algorithmic breakthrough.

Recently, there has been growing interest in solving the maximum $s$-$t$ flow problem in the challenging \emph{dynamic} setting, where the goal is to maintain a flow (or its value) under edge insertions and deletions and answer queries about the maintained flow efficiently.
For directed graphs, there are strong conditional hardness results~\cite{Dah16lbs, HKNS15omv} showing a lower bound of $\Omega(n)$ and $\Omega(\sqrt{m})$ amortized update time for maintaining \emph{exact} maximum flow, assuming the OMv conjecture. These strong polynomial lower bounds have motivated a shift in focus toward maintaining \emph{approximate} maximum flows.

Research on maintaining approximate maximum flows in the dynamic setting can be categorized into the following three lines of works. The first line of work~\cite{CGH20dynjtree, GRST21xhierarchy, vdBCK24decremental} deals with the fully dynamic setting, supporting both edge insertions and deletions and is based on dynamically maintaining tree-based cut approximations. However, these algorithms require at least a logarithmic loss in the quality of the maintained flow.
The second line of work~\cite{Hen97vertexconnectivity, GK21misflowmatch, GH23incflow} is purely combinatorial and works by repeatedly finding augmenting paths in an incremental residual graph together with a lazy rebuilding technique. While these algorithms achieve sub-linear update time, it is not clear how to use them to go beyond the $\sqrt{n}$ barrier on the update time.
To address this challenge, the third line of work~\cite{vdBLS23incmaxflow, vdBCK24incflow, CKL24incremental, vdBCK24decremental} leverages continuous optimization techniques to maintain a $(1-\eps)$-approximate max flow in the incremental (only edge insertions) and decremental (only edge deletions) settings in $m^{o(1)}$ update time. \cref{tab:dynamic_max_flow} offers a detailed summary of the state-of-the-art results on dynamic maximum flow algorithms.

Focusing on the $(1-\eps)$-approximation regime, the recent partially dynamic maximum flow algorithms suffer from the following two main drawbacks: (i) they depend on the powerful yet intricate machinery of continuous optimization methods, such as interior point methods, monotone multiplicative weight updates, and dynamic min-ratio cycle/cut problems, (ii) all existing algorithms incur an $m^{o(1)}$ factor in the update time, a limitation that arises in many modern dynamic graph-based data structures and appears difficult to overcome. This leads to the following fundamental question:

\begin{tcolorbox}
\begin{center}
Is there an incremental maximum flow algorithm that achieves $(1-\epsilon)$ approximation with polylogarithmic update time?
\end{center}
\end{tcolorbox}

We answer this question in the affirmative for dense graphs that are undirected and uncapacitated, as summarized in the theorem below.

\begin{restatable}{theorem}{incapxflow}
\label{thm:inc_apx_flow}
Given any $\eps \in (0,1)$, there is an incremental randomized algorithm against output-adaptive adversaries that maintains a $(1-\eps)$-approximate maximum $s$-$t$ flow $f$ under edge insertions on an undirected uncapacitated $n$-vertex graph $G$ with high probability in total time $O(m \log(n) \alpha(n) + n F^* \log^3 (n)/\eps)$, where $m$ is the number of edge insertions, $F^*$ is the value of the max flow after $m$ insertions, and $\alpha(n)$ is the inverse Ackermann function.
\end{restatable}

Note that in addition to dense graphs ($m=\Omega(n^2)$), our algorithm achieves polylogarithmic amortized update time even for graphs with small flow value $F^* = \tO(m/n)$, regardless of their density. %
An important feature of our algorithm is that it builds upon arguably simple and classic combinatorial techniques, such as residual graph sparsification and cut sparsification. Our algorithm can also be extended to allow non-simple graphs, which admit parallel edges between the same vertices.

\begin{table}[t]
\setlength{\tabcolsep}{5.2pt}
\small
\centering
\begin{tabular}{lcccccl}
\toprule
\textbf{Setting} & \textbf{Apx. Factor} & \textbf{Flow/Value} & \textbf{Directed} & \textbf{Weighted} & \textbf{Update Time} & \textbf{Reference} \\
\midrule
\multirow{7}{*}{\textbf{Incremental}}
& 1 & Flow & Yes & Yes & $O(F^*)$ & \cite{Hen97vertexconnectivity, GK21misflowmatch} \\
& 1 & Flow & No & No & $\tO(n^{2.5}m^{-1})$ & \cite{GH23incflow} \\
& $1+\eps$ & Flow & Yes & No & $\hO(m^{0.5}\eps^{-0.5})$ & \cite{GH23incflow} \\
& $1+\eps$ & Flow & Yes & Yes & $\hO(n^{0.5}\eps^{-1})$ & \cite{vdBLS23incmaxflow} \\
& $1+\eps$ & Flow & No & Yes & $O(m^{o(1)}\eps^{-3})$ & \cite{vdBCK24incflow} \\
& $1+\eps$ & Flow & Yes & Yes & $O(m^{o(1)}\eps^{-1})$ & \cite{CKL24incremental} \\
& $1+\eps$ & Flow & No & No & $\tO(nF^* m^{-1} \eps^{-1})$ & \cref{thm:inc_apx_flow} \\
\cmidrule{1-7}
\multirow{1}{*}{\textbf{Decremental}}
& $1+\epsilon$ & Value & Yes & Yes & $O(m^{o(1)}\epsilon^{-1})$ & \cite{vdBCK24decremental} \\
\cmidrule{1-7}
\multirow{3}{*}{\textbf{Fully dynamic}}
& $\tO(\log n)$ & Value & No & Yes & $\hO(n^{0.667})$ & \cite{CGH20dynjtree} \\
& $m^{o(1)}$ & Value & No & No & $m^{o(1)}$ & \cite{GRST21xhierarchy} \\
& $m^{o(1)}$ & Value & No & Yes & $m^{o(1)}$ & \cite{vdBCK24decremental} \\
\bottomrule
\end{tabular}
\caption{Results on dynamic max flow. The ``Flow/Value'' column indicates whether the algorithm maintains an actual flow or just the value of a flow. The update time stated is the amortized time over all updates.}
\label{tab:dynamic_max_flow}
\end{table}
\subsection{Technical Overview}
\label{sec:tech_overview}

The main idea underlying our result is to dynamize the static algorithm by Karger and Levine~\cite{KL15sampling} (abbrv. KL algorithm) that computes an \emph{exact} max flow in $\tO(m + nF^*)$ time. We start by reviewing their static construction and then identify the challenges in the dynamic setting, along with our approach to overcome them. We present the algorithm and its analysis in full in \cref{sec:inc_flow}.

The KL algorithm proceeds by repeatedly sampling $\rho = O(n \log n)$ edges from the residual graph, searching for an augmenting path in this sample, and then augmenting along this path before resampling again. When the search fails to find an augmenting path, the number of sampled edges is doubled to $2\rho$ and this process is repeated until the sampled graph becomes as large as the original graph.
Denoting the value of the maximum flow by $F^*$, they show that at least $F^*/2$ of the augmentations are performed on the sparsest graph with $\rho$ edges, at least $F^*/4$ of the augmentations are performed on the graph with $2\rho$ edges and so on.
After spending $\tO(m)$ on pre-processing to compute the sampling probabilities, the total running time thus adds up to $ \approx \sum_i (F^* /2^i) \cdot 2^{i-1} \rho$, which gives the required $\tO(m + nF^*)$ bound. Crucially, the last augmentation from $F^*-1$ to $F^*$ is performed on a graph of size $\Omega(m)$.

This final augmentation, which requires a sample of size $\Omega(m)$, effectively invalidates any attempts to use the KL approach for obtaining an incremental \emph{exact} max flow algorithm. To understand why, observe that after every edge insertion, the algorithm needs to check if this edge insertion creates a new $s\to t$ augmenting path in the sampled graph. As we discuss below, this can be done in two ways, but both lead to algorithmic dead ends.

We could try to use an incremental single source reachability data structure (e.g., the one by Italiano~\cite{Ita86incdirected}) to detect augmenting paths in the sampled graph, but since the actual augmentation reverses the direction of edges, we would need to reinitialize the incremental data structure on the sample after each augmentation. Since $F^*$ augmentations are performed on samples of size $\Omega(m)$, this leads to a time bound of $\Omega(mF^*)$.
As a concrete example, consider an initially empty bipartite graph $G = (S \cup T, \emptyset)$, with $s \in S$ and $t \in T$. In the first phase, insert all edges between $S \times T \setminus \{ t \}$, then in the second phase, insert all the remaining edges of the type $\{(v, t)  \}_{v \in S}$. Each edge insertion in the second phase increases the max flow by $1$, and any algorithm based on a KL-type approach with an incremental reachability data structure needs to check an $\Omega(m)$-edge graph to perform each augmentation.

The second approach would be to use a \emph{fully dynamic} directed $s\to t$ reachability data structure. This would solve the above problem, as  an augmentation in the residual graph can be simulated using $O(n)$ directed edge deletions and insertions. However, fully dynamic $s \to t$ reachability is known to admit strong conditional lower bounds -- for any $\eta > 0$, no algorithm can achieve $O(n^{1-\eta})$ worst-case update time and $O(n^{2-\eta})$ worst-case query time, assuming the OMv conjecture~\cite{HKNS15omv}.

Interestingly, we show that the first approach can be extended to the incremental setting if we relax our algorithm to maintain a \emph{$(1-\eps)$-approximate} $s$-$t$ max flow instead of an exact one.
The high-level reason why is as follows: To show that at least $F^*/2^i$ of the augmentations are found in the sampled graph of size $2^{i-1} \rho$, Karger and Levine prove that when at least $F^*/2^i$ of the max flow still remains to be augmented, a sampled graph of size $2^{i-1}\rho $ suffices to preserve the existence of an augmenting path with high probability\footnote{For us, high probability refers to the failure probability being $\le 4/n$.}. For our setting, if we want to maintain a $(1-\eps)$-approximate max flow, the above claim with $i = \log (1/\eps)$ tells us that it suffices to maintain a sparsifier with just $O(n \log (n) / \eps)$ edges incrementally.
This then removes the need to augment on an $\Omega(m)$-edge graph which lets us circumvent the above barrier.
While incrementally maintaining the sparsifier used by Karger and Levine would be highly non-trivial, we show that we can both efficiently maintain a different sparsifier incrementally, and that the new sparsifier is powerful enough to recover the guarantees required for the KL approach to work (at the cost of a $\log n$ factor in the sparsifier size).

The challenge of incrementally maintaining the KL sparsifier is due to its fragility under edge insertions. The KL sampling depends on a parameter called the \emph{strong connectivity} of an edge~\cite{BK15sampling}, and in the incremental setting, a single edge insertion could modify the strong connectivities (and thus the sampling probabilities) of \emph{all} other edges of the graph. As a result, no existing algorithm can even just maintain the sampling probabilities, let alone maintain an actual sample from the corresponding distribution. The same issue arises for other well-known connectivity measures used for sampling such as edge-connectivity~\cite{FHHP19sparsification} and effective resistances~\cite{SS11sparse}, and even maintaining approximations to them is quite involved~\cite{DurfeeGGP18, DGGP19schur, CGH20dynjtree, GaoLP21}. Instead, our key observation is to use \emph{Nagamochi-Ibaraki (NI) indices}~\cite{NI92sparsify} for our sampling because of their resilience to incremental updates.

While the other connectivity parameters mentioned are unique for an edge, the NI index is determined by the forest packing constructed, and different packings gives rise to different indices. This flexibility in choosing the index allows us to maintain the same value even under edge updates, which we crucially use to maintain the sampling probabilities. The resilience property we use is as follows: when inserting an edge, the NI index of all other edges remains the same, and our task is to only determine the NI index of the newly inserted edge efficiently (which then does not change in the future). Efficient maintenance of NI indices is discussed in \cref{sec:inc_ni}.

While we now maintain a set of sampling probabilities using NI indices, it is unclear how they can be used to replace strong connectivity in the KL algorithm. We extend the general result of Fung et al.~\cite{FHHP19sparsification} on using a wide variety of connectivity parameters for cut sparsification on \emph{undirected} graphs to showing that the same parameters also lead to cut sparsification on balanced directed graphs\footnote{In the static setting, Cen et al.~\cite{CCPS21cutbalance} have shown that sampling according to strong connectivity leads to cut sparsification of balanced directed graphs. Since we work in the dynamic setting, we require the guarantee to hold for a larger class of connectivity parameters, including NI indices.}, which cover residual graphs that arise in our algorithm as a special case. Their results sample each edge independently, which in our setting leads to a prohibitive $\Omega(m)$ time for sampling. We adapt this to repeated sampling from a probability distribution, where we only spend $O(\log n)$ time per edge for a total $O(n \log^2(n)/\eps)$ edges sampled. We present this result in full detail in \cref{sec:balanced_sparsification}.

Finally, we simplify our algorithm by not maintaining an exact sample from the distribution at each time step, but an oversample. Specifically, we obtain a sample of the correct size ($O(n \log^2 (n)/\eps)$) right after an augmentation, but between two augmentations, we add edges directly to the sample instead of trying to maintain the distribution. This could blow up the number of edges in the sample to $\Omega(m)$ at some time steps, but since each edge is added this way to at most one sample throughout the algorithm, the total time bound of $\tO(m + nF^*/\eps)$ still holds. This also helps us obtain our guarantee against adversaries that know the current flow maintained by the algorithm, since the sample remains independent of the flow: any time the flow is updated, we resample the graph based on NI indices (which are independent of the flow), and between two augmentations, all inserted edges are directly added to the sample.

\section{Preliminaries}
\label{sec:prelim}

\subsection{Graphs and flows}
\label{sec:graphsflows}

\paragraph*{Flows}
Our algorithmic results are on undirected uncapacitated graphs $G=(V,E)$. We denote $n = |V|$ and $m = |E|$.
For two vertices $s, t \in V$, an $s$-$t$ flow $f \in \R^m$ assigns a value $f_e$ to each edge such that $\sum_{e \sim v} f_e = 0$ for all $v \neq s, t$ with $|f_e| \le 1$ $\forall e$. The value of a flow $f$ is $F(f) = \sum_{e \sim s} f_e$, and the maximum flow $f^* = \argmax_{\text{$s$-$t$ flows $f$}} F(f)$, with value $F^* = F(f^*)$. We use $F$ without the argument $f$ when the flow is clear from context. For any $\alpha \le 1$, an $s$-$t$ flow $f$ is an $\alpha$-approximate flow if $F \ge \alpha \cdot F^*$. If $\alpha > 1$, then the statement holds with $\alpha' = 1/\alpha$.

\paragraph*{Model}
In the incremental model, we start with an empty graph with no edges, and the graph is revealed as a sequence of edge insertions $e_1, e_2, \ldots$, leading to graphs $G_1, G_2, \ldots$. Our goal is to maintain, after each edge insertion $e_i$, an $s$-$t$ flow $f$ that is a $(1-\eps)$ approximation to the max $s$-$t$ flow $f^*_i$ in the current graph $G_i$. We use $F_i^*$ to denote the value of the flow $f_i^*$, and we use $F^*$ to denote the final max $s$-$t$ flow value after all edge insertions.

\paragraph*{Directed graphs}
We denote by $\vec{G} = (V, \vec{E})$ a directed (multi-)graph without edge capacities. An edge of integer capacity $x \in \mathbb Z^{+}$ is represented using $x$ parallel edges. For any subset of vertices $S \subseteq V$, $\bar{S} = V \setminus S$ is the complement of $S$, and $\partial_{\vec{G}}(S)$ denotes the set of edges leaving $S$, which is also denoted by $\vec{C}(S) = \partial_{\vec{G}}(S)$. Similarly, we use $\cev{C}(S) = \partial_{\vec{G}}(\bar{S})$ to denote the set of edges entering $S$. We use $\vec{C}$ and $\cev{C}$ without the argument $S$ when the subset is clear from context. The capacity $u(\vec{C}) = |\vec{C}|$ of the cut $\vec{C}$ is the number of edges in $\vec{C}$.

\paragraph*{Underlying undirected graph}
For any directed (multi-)graph $\vec{G}$, we denote by $G$ the underlying undirected (multi-)graph obtained from $\vec{G}$ by forgetting edge orientations, and $\partial_G(S)$ is the set of edges leaving $S$, which is also denoted by $C(S) = \partial_G(S)$.
For a directed cut $\vec{C} = \partial_{\vec{G}}(S)$, we use $C = \partial_G(S)$ to denote the underlying undirected cut in $G$, and $u(C) = |C|$ to denote the number of edges in $C$.

\paragraph*{Residual graph}
For any undirected, uncapacitated graph $G$, the corresponding residual graph for the zero flow $f=0^m$ is given by $G_{0^m}$ where each edge $(u,v)$ in $G$ is replaced by two directed edges $\vec{uv}$ and $\vec{vu}$. For any $G$ and a non-zero flow $f$, the residual graph $G_f$ is constructed the following way: start with the graph $G_{0^m}$ as above, and for every edge $e$ that carries flow in the direction $u \to v$, replace the edge $\vec{uv}$ in $G_{0^m}$ with the edge $\vec{vu}$.
In that case there are two edges $\vec{vu}$ in $G_f$.
A flow $f^*$ is a maximum flow if and only if $G_{f^*}$ contains no $s \to t$ paths.

\subsection{Sampling parameters}
\label{sec:samplingparameters}

We define our importance parameter for sampling edges next.
Since the residual graph has multiple edges between the same vertex pairs, we consider multi-graphs (which allow parallel edges) for the following definition.

\fancydef[niindices][NI indices]{niindex}{NI index}
\begin{definition}[NI index~\cite{NI92sparsify}]
\label{def:niindex}
Given an undirected, unweighted (multi-)graph $G$, an ordered sequence of edge-disjoint spanning forests $T_1, T_2, \ldots$ of is said to be an NI forest packing of $G$ if each $T_i$ is a maximal spanning forest on $G \setminus \cup_{1 \le j < i} \{ T_j \}$, and $\cup_i \{  T_i \}$ is a partition of all the edges of $G$. The NI index $\ell_e$ of an edge $e$ is the index $i$ of the forest $T_i$ that $e$ belongs to.
\end{definition}

While the \niindices are the connectivity parameters that we will use for our sampling scheme due to their resilience to incremental graph updates, the following connectivity parameter is required for our directed sparsification proofs and cut-counting.

\fancydef{edgeconnectivity}{edge-connectivity}
\begin{definition}[Edge-connectivity]
\label{def:edgeconnectivity}
Given an undirected unweighted (multi-)graph $G$ and two vertices $u$ and $v$, the edge connectivity between $u$ and $v$ is defined to be the value of the minimum cut (and thus also the maximum flow) between $u$ and $v$ in $G$.
For an edge $e = (u,v)$, we define the \edgeconnectivity of $e$ to be the edge connectivity between $u$ and $v$.
\end{definition}

As we perform \emph{non-uniform} sampling of the edges, we maintain connectivity parameters $\lambda_e$ that influence the sampling probabilities.
The \niindex of an edge and the \edgeconnectivity of an edge are examples of such connectivity parameters.
We will show concentration bounds for classes of edges whose $\lambda_e$s are close to each other, which motivates the following definition.

\fancydef[connectivityclasses][connectivity classes]{connectivityclass}{connectivity class}
\begin{definition}[Connectivity class]
\label{def:connectivityclass}
Given an undirected graph $G$ and connectivity parameters $\lambda_e$ for each edge, taking $\Lambda = \max_e \left\lceil \log \lambda_e  \right\rceil + 1$,
the connectivity classes $\mathcal F = \{ F_i: 1 \le i \le \Lambda \}$ are given by
\[
    F_i = \{ e : 2^{i-1} \le \lambda_e \le 2^i-1 \}.
\]
\end{definition}

For any \connectivityclass $F_i$, the graph $(V, F_i)$ does not have the same \edgeconnectivity properties that $G = (V, E)$ did. An edge $e \in F_i$ could have \edgeconnectivity $\Omega(n)$ in $G$ and $1$ in $(V, F_i)$, e.g., a graph with a single $(s,t)$ edge and $\Omega(n)$ parallel paths from $s$ to $t$ of length two. The following definitions from \cite{FHHP19sparsification} are used to construct a subgraph of $G$ that is slightly larger than $F_i$ that certify connectivity properties of edges in $F_i$.

\fancydef{decomposition}{decomposition}
\fancydef{piconnectivity}{$\Pi$-connectivity}
\begin{definition}[$\Pi$-connected decomposition~\cite{FHHP19sparsification}]
\label{def:decomposition}
Given an undirected graph $G$ and connectivity parameters $\lambda_e$, a decomposition $\mathcal G = \left( G_i = (V, E_i) \right)_{1 \le i \le \Lambda}$ is a sequence of subgraphs of $G$ that satisfy $F_i \subseteq E_i$ for all $i$ with $\Lambda = \max_e \left\lceil \log \lambda_e \right\rceil + 1$.
\label{def:piconnectivity}
Further, for a sequence of parameters $\Pi = \left( \pi_i \right)_{1 \le i \le \Lambda}$,
the decomposition satisfies $\Pi$-connectivity if
all edges $e \in F_i$ have edge-connectivity at least $\pi_i$ in $G_i$ for all $i$.
\end{definition}

Note that the same edge is allowed to appear in multiple sets $E_i$, and the decomposition is \emph{not} edge disjoint. This property is crucial in our proofs.

\fancydef{gammaoverlap}{$\gamma$-overlap}
\begin{definition}[$\gamma$-overlap~\cite{FHHP19sparsification}]
\label{def:gammaoverlap}
For any $\gamma \ge 1$,
an undirected graph $G$ and its $\Pi$-connected \decomposition $\mathcal G$ satisfies $\gamma$-overlap
if for all $i$ and for all cuts $C = \partial_{G}(S)$,
\[
\sum_{i = 0}^\Lambda u(C \cap E_i) \cdot \frac{2^{i-1}}{\pi_i} \le \gamma \cdot u(C)
\]
\end{definition}

\subsection{Balance and cut counting}
\label{sec:balancecutcounting}

\paragraph*{Imbalance of directed graphs}
Preserving various graph properties by subsampling is difficult on general directed graphs. However, sampling algorithms can be successfully applied when the directed graph exhibits some nice structure, for example, they have a similar number of edges crossing a cut in both directions. Eulerian graphs have exactly the same number of edges crossing any cut in both directions, and more pertinently, the residual graph of a flow that is not a $(1-\eps)$ approximate max flow has at least an $O(\eps)$ fraction of edges crossing any cut in one direction over the other.

\fancydef{balance}{balance}
\begin{definition}[balance~\cite{EMPS16balance}]
\label{def:balance}
Given a directed (multi-)graph $\vec{G} = (V, \vec{E})$, the (im)balance of a cut $\vec{C} = \partial_{\vec{G}}(S)$ is defined as
\[
\beta(\vec{C})
= \frac{u(\vec{C})}{u(\cev{C})}
\] where $\cev{C} = \partial_{\vec{G}}(\bar{S})$ is the set of edges in the other direction of the cut.
\end{definition}

\paragraph*{Counting cut projections}
Karger shows that there are at most $n^{2\alpha}$ distinct cuts whose size is $\le \alpha \lambda$,
where $\lambda$ is the global mincut. We, like Fung et al.~\cite{FHHP19sparsification}, look at larger cuts in the graph and require a tighter bound than the one given by Karger, which requires the following definition.

\fancydef{kprojection}{$k$-projections}
\begin{definition}[directed $k$-projection]
\label{def:kprojection}
Let $\vec{G} = (V, \vec{E})$ be a directed (multi-)graph and let $G$ be the corresponding undirected graph. The $k$-projection of any cut $\vec{C} = \partial(S)$ in $\vec{G}$ is $\vec{C} \cap H_k$ where
\[
H_k = \left\{ \vec e \in \vec E: \text{the \edgeconnectivity of $e$ in $G$ is at least $k$} \right\}
\]
\end{definition}

We use this definition to extend the following statement to directed graphs, which was originally shown by Fung et al.~\cite{FHHP19sparsification} for undirected graphs.

\begin{restatable}[\cite{FHHP19sparsification}]{lemma}{fhhpprojection}
\label{lem:fhhpprojection}
Let $G = (V, E)$ be an undirected (multi-)graph. Let $k \ge \lambda$ be any real number, where $\lambda$ is the value of the global minimum cut in $G$. Consider all the cuts $C = \partial_G(S)$ in $G$ such that $u(C) \le \alpha k$, and let $\mathcal K^{\downarrow k}_{\alpha}$ be the set of all undirected $k$-projections of these cuts, where the undirected $k$-projection of a cut $C = \partial_G(S)$ is $C \cap H_k$ and $H_k$ is the set of edges with edge-connectivity at least $k$. Then $|\mathcal K^{\downarrow k}_\alpha| \le n^{2\alpha}$.
\end{restatable}

\subsection{Data Structures}
\label{sec:datastructures}

We also need three classic data structures for our results, which we produce here.

\setlist{nosep,topsep= 0.3em-\parskip}

\fancydef{unionfind}{\textsc{UnionFind}}
\begin{fact}[\textsc{UnionFind}~\cite{TvL84unionfind}]
\label{def:unionfind}
There is a data structure that supports the following operations
\begin{itemize}
    \item $\fn{Add}(u)$: Adds $\{u\}$ to the set of sets in $O(1)$ time.
    \item $\fn{Find}(u)$: Finds the representative of the set that $u$ belongs to in amortized $O(\alpha(n))$ time.
    \item $\fn{Union}(u, v)$: Replaces the sets that $u$ and $v$ belong to with their union in $O(1)$ time.
\end{itemize}
where $n$ is the number of elements present at that time, and $\alpha(n)$ is the inverse Ackermann function.
\end{fact}

\fancydef{SSR}{\textsc{SingleSourceReachability}}
\begin{fact}[\textsc{SingleSourceReachability}~\cite{Ita86incdirected}]
\label{def:SSR}
There is a data structure that supports the following operations on a graph $\vec{G}$ with $n$ vertices
\begin{itemize}
    \item $\fn{Initialize}(\vec{G}, s)$: Initializes a data structure on $\vec{G}$ with special vertex $s$ in $O(m+n)$ time.
    \item $\fn{Insert}(e = \vec{uv})$: Inserts the directed edge $\{e = \vec{uv}\}$ into $\vec{G}$ in amortized $O(1)$ time.
    \item $\fn{Reachable}(t)$: Returns ``True'' if $t$ is reachable from $s$ in $\vec{G}$ in amortized $O(1)$ time.
    \item $\fn{GetPath}(t)$: Returns a directed $s \to t$ path in $\vec{G}$ in time proportional to the length of the path.
\end{itemize}
\end{fact}

\fancydef{BST}{\textsc{BinarySearchTree}}
\begin{fact}[\textsc{BinarySearchTree}]
\label{def:BST}
There is a data structure that supports the following operations
\begin{itemize}
    \item $\fn{Initialize}$: Initializes the binary search tree data structure.
    \item $\fn{Insert}(e, x)$: Inserts the key $e$ with value $x$ in amortized $O(\log n)$ time.
    \item $\fn{Search}(x)$: Returns the largest key $e$ with value $\le x$ in amortized $O(\log n)$ time.
\end{itemize}
where $n$ is the number of items inserted into the data structure at time of operation.
\end{fact}

\paragraph*{High probability bounds}
Finally, we say that an algorithm satisfies a guarantee \emph{with high probability} if the guarantee holds with probability at least $1 - 4/n$.

\paragraph*{Adversaries}
An adversary is said to be \emph{oblivious} if it fixes the update sequence before the algorithm picks its random bits. An adversary is \emph{output-adaptive} if it is allowed to choose the next update/query based on the past updates and the sequence of answers given by the algorithm to past queries. An adversary is \emph{state-adaptive} if it is output-adaptive and also allowed access to the internal state of the algorithm before choosing the next update.
\section{Incremental Maximum Flow}
\label{sec:inc_flow}

We assume that the algorithm is started on the empty graph.
Our main result in this section is the following:
\incapxflow*

\begin{algorithm}[ht]
\DontPrintSemicolon
\caption{Algorithm for incremental approximate maximum flow on undirected, uncapacitated graphs}
\label{alg:inc_flow}
\Fn{\fn{Initialize}($\eps$)}{
$\rho \gets \samplesize$\;
$f \gets \emptyset$, $F \gets 0$, $G_f \gets (V, \emptyset)$, $H \gets G_f$ \;
$D \gets \SSR.\fn{Initialize}(H, s)$\;
$K \gets$ \hyperref[alg:inc_ni_sampling]{\textcolor{black}{\fn{IncNISample}}}.\fn{Initialize} \algcomment{\cref{alg:inc_ni_sampling}}
}
\Fn{\fn{Insert}($e=(u,v)$)}{
    Add $\vec{uv}$ and $\vec{vu}$ to $H$ and $D$\;
    $K.\fn{Insert}(e)$\;
    \If{$D.\fn{Reachable}(t)$} {
        $p \gets D.\fn{GetPath}(t)$\;
        update $f$ and $G_f$ by augmenting along $p$, and update $F \gets F+1$\;
        reset $H \gets (V, \emptyset)$\;
        \For{$i = 1, 2, \ldots, \rho$} {
            $(u,v) \gets K.\fn{Sample}$\;
            \For{$\vec{ab} \in \{ \vec{uv}, \vec{vu} \}$ \label{line:bothdirstart}} {
                \If{$\vec{ab} \in G_f$} {
                    Add $\vec{ab}$ to $H$ \label{line:bothdirend}\;
                }
            }
        }
        reinitialize $D \gets \SSR.\fn{Initialize}(H, s)$
    }
}
\Fn{\fn{Query}(flow/value)}{
\textbf{if} query = flow \textbf{then return} $f$ \textbf{else return} $F$
}
\end{algorithm}

We present the algorithm achieving the guarantees in \cref{alg:inc_flow}. The algorithm works in phases. At the start of a new phase, we sample $O(n \log^2 (n) / \eps)$ edges from the residual graph using a specific probability distribution which we will discuss later, and add them to the \emph{sampled graph} $H$. On top of $H$, we run an incremental directed single-source reachability algorithm, called $D$, starting from $s$. At the beginning, and after each edge insertion, we query $D$ if there exists an $s \to t$ path in the sample.
As long as $D$ does not return ``Yes'', for each edge insertion $e = (u,v)$, we add edges $\vec{uv}$ and $\vec{vu}$ to $D$. When $D$ returns ``Yes'', this corresponds to an $s \to t$ path in the residual graph, and we retrieve this path from $D$, augment along this path in $G_f$, and start a new phase.

We need to argue about the time bounds and the correctness. For correctness, we show that whenever an edge insertion increases the max flow to a value larger than a $(1-\eps)^{-1}$ factor over the current flow maintained by the algorithm, the existence of an $s \to t$ path in the sampled graph $H$ is guaranteed with high probability. Since $H$ is a subsample of the residual graph $G_f$, this provides us with an augmenting path in $G_f$ along which we augment the flow.
With a union bound over the at most $n$ times when this can happen (since the max flow value is at most $n$ and each augmentation increases the flow by $1$), this will imply that at the end of every time step, the algorithm maintains a valid $(1-\eps)$-approximate $s$-$t$ flow with high probability.
To show this claim, we first show that an extension of Fung et al.'s high probability result~\cite{FHHP19sparsification} on independent sampling for cut sparsification on undirected graphs also extends to repeated sampling for cut sparsification on balanced directed graphs. When combined with the fact that the residual graph of a flow that is not a $(1-\eps)$-approximate max flow is $\Omega(\eps)$-balanced, this shows that sampling approximately $n \log^2(n)/\eps$ edges based on the probability distribution discussed below will preserve directed cuts, and thus also the existence of an augmenting path in residual graphs of non-$(1-\eps)$-approximate flows. Our initial sample at the beginning of a phase already satisfies the size requirement above, and since we add further edges directly to the maintained sample, we always oversample the rate required for the augmenting path preservation guarantee.

For the time bounds, assuming that the sampling can be done in time approximately $O(\log n)$ per sample and that the reachability data structure runs in total time proportional to the number of edges in $D$, we will show that the algorithm runs in $\tO(m + nF^*/\eps)$ time.
Note that the graph $H$ is resampled at most $F^*$ times, since each resample occurs after an augmenting path has been found in the residual graph. Each resample adds $\tO(n/\eps)$ edges to $D$, for a total of $\tO(nF^*/\eps)$ sampled edges inserted to $D$ throughout the algorithm. This is in addition to the $m$ edges which are inserted to $D$ directly on insertion to $G$, which gives the claimed running time bound.

Finally, we maintain \niindices incrementally for the probability distribution, which we discuss in \cref{sec:inc_ni}.
We use the following definition in our proofs.

\fancydef[phases][Phases]{phase}{Phase}
\begin{definition}[Phase]
\label{def:phase}
Let $t_0 = 0$, and let $t_i$ be the time step when the $i^{\mathrm{th}}$ augmenting path is found. Let $m_k$ be the number of edges inserted in time steps $(t_{k-1}, t_k]$. Phase $k$ refers to the computations performed after the $(k-1)^{th}$ augmentation finishes until the $k^{th}$ augmentation ends.
\end{definition}

We also use the following two statements directly, whose proofs can be found in \cref{sec:inc_ni} and \cref{sec:balanced_sparsification} respectively.

\begin{restatable}{lemma}{NIsample}
\label{lem:NI_sample}
There is an incremental data structure that deterministically maintains an NI forest packing (and thus the \niindex $\ell_e$ of each edge) of an undirected graph under edge insertions in total time $O( m \cdot \log m \cdot \alpha(n))$ where $\alpha(n)$ is the inverse Ackermann function and $m$ is the number of edges in the final graph. At any point of time, it also allows querying for an edge sampled from the probability distribution $\left\{ \ell_e^{-1}/L \right\}_{e \in E}$ where $L = \sum_e \ell_e^{-1}$ in amortized time $O(\log n)$ for each sample.
\end{restatable}

\begin{algorithm}[t]
\DontPrintSemicolon
\caption{Algorithm for directed balanced sparsification}
\label{alg:sample}
\KwInput{Directed graph $\vec{G} = (V, \vec{E})$, connectivity parameters $\{\lambda_e\}_{e \in E}$, parameters $\beta, \gamma \ge 1$, and $\eps < 1$}
\KwOutput{Sparsified graph $\vec{H} = (V, \vec{F}, w)$}
$\rho \gets \generalsamplesize$\;
$\vec{H} \gets (V, \emptyset, w)$\;
\For{$i \in 1, 2, \ldots, \rho$}{
    Sample an edge $e$ from the probability distribution $\{ p_e = \lambda_e^{-1} / \sum_e \lambda_e^{-1} \}$\;
    Insert $e$ into $\vec{H}$ with weight $w_e = (\rho \cdot p_e)^{-1}$ (additively increasing its weight if $e$ already exists in $\vec{H}$)
}
\end{algorithm}

\begin{restatable}{theorem}{generalsparsification}
\label{thm:generalsparsification}
Let $\vec{G} = (V, \vec{E})$ be any directed (multi-)graph, and $G$ be the underlying undirected graph. Let $\{\lambda_e\}_{e \in E}$ be some connectivity parameters in $G$, and $\beta, \gamma \ge 1$ and $\eps < 1$ be input parameters. Let $\vec{H} = (V, \vec{F}, w)$ be the random graph with $O\!\left(\gamma \beta \ln (n) \cdot \sum_{e \in E} \lambda_e^{-1} / \eps^2 \right)$ edges generated as in \cref{alg:sample}.

If there exists a $\Pi$-connected decomposition $\mathcal G = \{ G_i = (V, E_i): 1 \le i \le \Lambda \}$ of $G$ that satisfies \gammaoverlap, then for all cuts $\vec{C} = \partial_{\vec{G}}(S)$ of \balance at least $\beta^{-1}$ in $\vec{G}$, it holds simultaneously with probability $\ge 1 - 4/n^2$ that
\[
(1-\eps) \cdot u(\vec{C}) \le w(\vec{C}) \le (1+\eps) \cdot u(\vec{C}).
\]
\end{restatable}

\begin{remark}
This theorem guarantees sparsification for all balanced cuts, even in potentially unbalanced graphs. When restricted to only using strong connectivity, this statement was shown by Cen et al.~\cite{CCPS21cutbalance} with stronger parameters where the quality of sparsification depends on the balance of the cut, who then used it to give a clearer proof of Karger and Levine~\cite{KL15sampling}'s static result.
\end{remark}

\subsection{Running time}
\label{sec:time}

\begin{restatable}{lemma}{phasetime}
\label{lem:phasetime}
\phase $k$ takes total time $O(m_k \log (n) \alpha(n) + n \log^3 (n) / \eps)$, where $m_k$ is the number of edges inserted in \phase $k$ and $\alpha(n)$ is the inverse Ackermann function.
\end{restatable}
\begin{proof}
Phase $k$ starts with sampling $\Osamplesize$ edges from the sampler, where each sample takes time $O(\log n)$ by \cref{lem:NI_sample}. $D$ has a total of $m_k + \Osamplesize$ many edges added to it in its entire lifetime, which gives a total update time of $O(m_k + n \log^2 (n)/\eps)$. Inserting each edge into the sampler takes amortized time $O(\log m \cdot \alpha(n))$ by \cref{lem:NI_sample}. Augmenting along a path takes time proportional to the length of the path, which is at most $n$.
\end{proof}

\begin{restatable}{lemma}{time}
\label{lem:time}
The algorithm takes total time $O(m \log (n) \alpha(n) + n \log^3 (n) F^*/\eps)$ where $F^*$ is the value of the final max flow after all edge insertions.
\end{restatable}
\begin{proof}
Let $K$ be the total number of phases in the algorithm. Then $K \le F^*$ since each phase increases the value of the flow by $1$.
Since the set of edges inserted in each phase are disjoint, $\sum_k m_k = m$.
By \cref{lem:phasetime}, phase $k$ takes time $O(m_k \log (n) \alpha(n) + n \log^3 (n)/\eps)$. Together, all phases take time
\[
\sum_{k = 1}^{K} O(m_k \log(n) \alpha(n) + n\log^3(n)/\eps) = O\left(m \log (n) \alpha(n) + n \log^3 (n) F^*/\eps\right) \qedhere
\]
\end{proof}

\subsection{Correctness}
\label{sec:correctness}

We first show that the residual graph of a non-$(1-\eps)$-approximate max flow is $\eps/2$-balanced. This then allows us to use our result on sparsifying balanced directed graphs.

\begin{restatable}{lemma}{residualcuts}
\label{lem:residualcuts}
Let $f$ be any $s$-$t$ flow in an undirected graph $G$ of value $F$, and let $G_f$ be its corresponding residual graph. Suppose $F \le (1-\eps) F^*$, where $F^*$ is the value of the maximum flow. Then every cut in $G_f$ is at least $\eps/2$-balanced.
\end{restatable}
\begin{proof}
For any cut $\vec{C} = \partial(S)$ that is not an $s$-$t$ cut (or a $t$-$s$ cut), the net flow out of $S$ in $G_f$ is zero, and thus $\vec{C}$ is $1$-balanced. All $t$-$s$ cuts are $\ge 1$-balanced since the residual graph has more flow in the $s \to t$ direction and thus more capacity in the $t \to s$ direction.
Consider an $s$-$t$ cut $\vec{C}$. Let $c$ be the capacity of the corresponding undirected cut $C$ in the graph $G$. Since $F$ units of flow cross $C$ in the forward direction $s \to t$, the forward capacity across $C$ in $G_f$ is $c - F$, and the backward capacity is thus $c + F$ since the total capacity is $2c$. Using (in order) $c \ge F$, $c \ge F^*$, and $F \le  (1-\eps) F^*$, we get
\[
\beta(\vec{C})
= \frac{u(\vec{C})}{u(\cev{C})}
= \frac{c-F}{c+F}
\ge \frac{c-F}{2c}
= \frac{1}{2} \cdot \left( 1 - \frac{F}{c} \right)
\ge \frac{1}{2} \cdot \left( 1 - \frac{F}{F^*} \right)
\ge \frac{\eps}{2}
\qedhere
\]
\end{proof}

We then show that the sum of inverse NI indices is not too large.

\begin{restatable}{lemma}{NIinversesum}
\label{lem:NIinversesum}
Let $\left\{ \ell_e \right\}_{e \in E}$ be a set of NI indices for an undirected (multi-)graph $G$ with $m$ edges. Then $\sum_e \ell_e^{-1} \le 2 n \log m$.
\end{restatable}
\begin{proof}
Consider the NI forest packing $T_1, T_2, \ldots$ corresponding to the NI indices $\ell_e$. Since each forest contains at least one edge, the total number of forests is at most $m$. As each edge $e$ in $T_i$ has $\ell_e = i$, the sum of NI indices of all edges in forest $i$ contributes at most $\frac{n-1}{i}$ to the total sum. Thus,
\[
\sum_{e \in E} \ell_e^{-1}
= \sum_{1 \le i \le m} \sum_{e \in T_i} \ell_e^{-1}
\le \sum_{1 \le i \le m} \frac{n-1}{i}
\le (n-1) \cdot (2 \log m)
\le 2 n \log m
\]
where $ 1 + \frac{1}{2} + \frac{1}{3} + \dots + \frac{1}{m}$ is the $m$-th harmonic number,
which is upper bounded by $2 \log m$.
\end{proof}

Since we work with the residual graph which has two copies of each edge of $G$, we prove the following lemma to show that an NI forest packing of $G$ can be easily transformed into an NI forest packing of the undirected residual graph.

\begin{restatable}{lemma}{niresidual}
\label{lem:niresidual}
Let $\{ \ell_e \}_{e \in E}$ be a set of NI indices for an undirected, simple graph $G = (V, E)$. Let $E'$ be a copy of the edges in $E$. Then, for the multigraph $H = (V, E \cup E')$ which doubles the edges in $E$, assigning $h_e = 2 \ell_e - 1$ and $h_{e'} = 2 \ell_e$ is a set of NI indices for $H$.
\end{restatable}
\begin{proof}
If $T_1, T_2, \ldots$ forms an NI forest packing of $G$, then $T_1, T_1, T_2, T_2, \ldots$ forms an NI forest packing of $H$. Letting $e$ be the edge in the first copy of $T_i$ and $e'$ be the edge in the second copy gives the lemma.
\end{proof}

We use the result of \cite{FHHP19sparsification} that shows that an NI forest packing can be used to obtain a \decomposition that satisfies \piconnectivity and \gammaoverlap.

\begin{restatable}[\cite{FHHP19sparsification}]{lemma}{NIdecomposition}
\label{lem:NIdecomposition}
Let $T_1, T_2, \ldots$ be an NI forest packing of an undirected graph $G$ leading to \niindices $\ell_e$, and let $\mathcal F = \{ F_i \}$ be the corresponding \connectivityclasses. Then $E_1 = F_1$ and $E_i = F_{i-1} \cup F_i$ for $i \ge 2$ is a \decomposition that satisfies \piconnectivity and \gammaoverlap for $\pi_i = 2^{i-1}$ and $\gamma = 2$, i.e.,
 for all $i$ and for all cuts $C = \partial_{G}(S)$,
\[
\sum_{i = 0}^\Lambda u(C \cap E_i) \le 2 \cdot u(C).
\]
\end{restatable}

Using these, we will first show that sampling $\Omega(n \log^2(n)/\eps)$ edges suffices to sample at least one edge across all balanced cuts.

\begin{restatable}{lemma}{nisampling}
\label{lem:nisampling}
Consider a simple undirected, uncapacitated graph $G$ with some $s$-$t$ flow $f$ of value $F$. Let $\ell_e$ be a \niindex of edge $e$. Let $H' = (V, \vec{F})$ be a sample of $\samplesize$ edges chosen i.i.d.~with probability $\left\{ \ell_e^{-1} / \sum_e \ell_e^{-1} \right\}_{e \in G_f}$ from the residual graph $G_f$.
Then it holds with probability $1-4/n^2$ that for every directed cut $\vec{C} = \partial(S)$ in $G_f$ of balance at least $\eps/2$, there is an edge $e \in \vec{C}$ that is sampled in $H'$.
In particular, if $F^*$ is the value of the max flow in $G$ and $F \le (1-\eps) F^*$,
then there exists an augmenting path in the sampled graph $H'$ with probability $1 - 4/n^2$.
\end{restatable}
\begin{proof}
We will first argue that sampling the mentioned number of edges and reweighting them as in \cref{alg:sample} gives a (weighted) cut sparsifier $\vec{H} = (V, \vec{F}, w)$ of $G_f$.
This implies positive weight across every balanced cut in $\vec{H}$.
Then, we note that the property of presence of an edge across the cut holds regardless of the presence of weights. Since the sample $H'$ chosen in the lemma is an unweighted version of $\vec{H}$, the lemma follows.

Since there are at most ${\binom n 2}$ edges in a simple graph $G$ (and thus at most $n^2$ edges in a residual graph), \cref{lem:NIinversesum} gives that that $\sum_{e \in E} \ell_e^{-1} \le 4 n \log n$.
By \cref{lem:NIdecomposition}, there is a \decomposition for NI indices that satisfies \piconnectivity and \gammaoverlap for $\pi_i = 2^{i-1}$ and $\gamma = 2$.
Thus, if one runs \cref{alg:sample} with $\rho = \samplesize$ edges, then by \cref{thm:generalsparsification} one obtains a weighted $(1+\eps')$-sparsifier $\vec{H} = (V, \vec{F}, w)$ of the residual graph $G_f$ where the parameters are set as $\eps' = 1/2$, $\gamma = 2$, $\beta = 2/\eps$, and $\sum_e \ell_e^{-1} \le 4 n \log n$. Thus, for all cuts $\vec{C}$ in $G_f$ of balance at least $\eps/2$~(\cref{lem:residualcuts}), we have that the weighted sparsifier $\vec{H} = (V, \vec{F}, w)$ from \cref{thm:generalsparsification} satisfies
\[
w(\vec{C}) \ge \frac{u(\vec{C})}{2} > 0
\]
simultaneously with probability $1 - 4/n^2$. In particular, this also means that at least one edge from $S$ to $\bar{S}$ is sampled in $\vec{H}$ for every directed $s$-$t$ cut $\partial(S)$.

Defining $H' = (V, \vec{F})$ to be the unweighted graph obtained from $\vec{H}$ by dropping the edge weights, the existence of an edge crossing $S$ to $\bar{S}$ holds in $H'$ as well for all cuts $\partial(S)$ with balance at least $\eps/2$ simultaneously, with probability $1-4/n^2$.
Since the sampled graph stated in the lemma is generated exactly as in $H'$, the lemma follows.
\end{proof}

\begin{remark}
While we discuss our algorithm purely for simple graphs, this is not an intrinsic barrier. We use the simpleness to apply Lemma~\ref{lem:NIinversesum} to Lemma~\ref{lem:nisampling}, where we concretely use $2 \log n$ as an upper bound on $\log m$, and for a union bound in \cref{lem:connectivityclassconc}. If we are given a different upper bound on $m$ beforehand, then we could use this value instead in the algorithm to obtain the same guarantees.
\end{remark}

This lemma gives us that right after a sample happens, and before any future edge insertions, the augmenting path guarantee holds with high probability.

\begin{restatable}{lemma}{beginningcorrectness}
\label{lem:beginningcorrectness}
At the beginning of any phase $k$, if $F \le (1-\eps)F^*_{t_{k-1}}$, then there exists an augmenting path in $H$ with probability $1 - 4/n^2$. Here, $F^*_i$ is the max flow at time $i$, and $t_{k-1}$ is the time step when the $(k-1)$-st augmentation ends. This guarantee holds against output-adaptive adversaries.
\end{restatable}
\begin{proof}
Recall that at the beginning of a phase, the graph $H$ is sampled according to \niindices in the undirected graph $G$.
We show that the conditions of \cref{lem:nisampling} hold, even with an output-adaptive adversary.
Note that the sampling \emph{probabilities} are maintained deterministically and are known even to an oblivious adversary.
While an output adaptive adversary knows the flow maintained by the algorithm (and thus the current residual graph as well), the sampling according to \niindices is \emph{independent} of the flow maintained by the algorithm and its previous outputs. \cref{lem:nisampling} shows the existence of an augmenting path for the residual graph of \emph{any} flow, so no matter which updates were performed in the past and what information the output-adaptive adversary has. Since we assume in the lemma that we are at the beginning of a phase, there have been no edge insertions after the sampling of $H$.
Since we sample an undirected edge and add both directions of that edge into $H$ (lines~\ref{line:bothdirstart} to \ref{line:bothdirend}), this is an oversample of the rate required by \cref{lem:nisampling}, which only requires adding one direction of the edge.
Thus, the claim follows by the guarantees of \cref{lem:nisampling}.
\end{proof}

To extend this guarantee to all time steps, we show that \cref{lem:nisampling} holds regardless of the future sequence of edge insertions that happen in the incremental algorithm. Thus, an adversary who cannot see the internal randomness of the sample does not get any power by choosing future updates, since the guarantee holds regardless of the choice of future edge insertions.
\begin{restatable}{lemma}{anyfutureinsertions}
\label{lem:anyfutureinsertions}
Consider an undirected, uncapacitated graph $G = (V, E)$ with some $s$-$t$ flow $f$ of value $F$.
Let $H' = (V, \vec{F})$ be the graph sampled from $G_f$ as in \cref{lem:nisampling}, and let $\mathcal E = \{ X \subseteq (V \times V) \mid X \cap E = \emptyset \}$ be the set of all possible future updates. For any $X \in \mathcal E$, define $F^*_X$ to be the max flow in the graph $G_X = (V, E \cup X)$, and define $H'_X = (V, \vec{F} \cup X)$. Then the following guarantee holds with probability at least $1 - 4/n^2$: for every $X \in \mathcal E$, if $F \le (1-\eps) F^*_X$, then there exists an augmenting path in $H'_X = (V, \vec{F} \cup X)$.
\end{restatable}
\begin{proof}
By \cref{lem:nisampling}, we have that with probability at least $1 - 4/n^2$, the initially sampled graph $H'$ of $G_f$ satisfies the following: for any directed $s$-$t$ cut $\partial(S)$ of balance at least $\eps/2$, there is at least one edge sampled from $S$ to $\bar{S}$.
We assume in what follows that $H'$ satisfies this guarantee, thereby conditioning on an event of probability at least $1 - 4/n^2$.
Fix an arbitrary $X \in \mathcal E$ such that $F \le (1-\eps) F^*_X$. We will show that the augmenting path guarantee holds in $H'_X$, which completes the proof.

Note that the augmenting path guarantee is equivalent to showing that for all directed $s$-$t$ cuts $\partial(S)$ in $H'_X$, there is at least one edge from $S$ to $\bar{S}$.
Let $\vec{C} = \partial(S)$ be an arbitrary directed $s$-$t$ cut.
By \cref{lem:residualcuts}, $\vec{C}$ has balance at least $\eps/2$ in $(G \cup X)_f$.
If $\vec{C}$ had balance at least $\eps/2$ even in $G_f$,
then there is at least one edge sampled across this cut in $H'$ by the conditioning above. If not, then the balance of $\vec{C}$ changed between $G_f$ and $(G \cup X)_f$,
which means that an edge was inserted across $\vec{C}$ in $X$ and, thus, the undirected edge is also added to $H'_X$. In either case, $\vec{C}$ has an edge crossing the cut in $H'_X$. Since the choice of cut was arbitrary, this proves the claim.
\end{proof}

\begin{restatable}{lemma}{singlecorrectness}
\label{lem:singlecorrectness}
Fix a phase $k$. Let $i$ be the first time step within the phase when $F < (1-\eps) F^*_i$, where $F^*_i$ is the value of the max flow at time $i$. Then there exists an augmenting path in $H$ with probability $1 - 4/n^2$. The guarantee holds even if an output-adaptive adversary inserts edges in this phase.
\end{restatable}
\begin{proof}
For phase $k$,
we use \cref{lem:anyfutureinsertions} to argue that the guarantee holds in this phase, for any sequence of future edge insertions.
In particular, let $i$ be the first time step within the phase when $F < (1-\eps) F^*_i$. Let $X$ be the sequence of edge insertions in this phase, the augmenting path guarantee holds in the graph $H'_X$ with probability $1-4/n^2$ regardless of the identity of $X$.
Since the graph $H'_X$ from \cref{lem:anyfutureinsertions} is the same as the sampled graph $H$ maintained by \cref{alg:inc_flow} after inserting the edges in $X$ in this phase, the lemma holds.
\end{proof}

Finally, we take a simple union bound to extend the guarantee to all time steps.

\begin{restatable}{lemma}{correctness}
\label{lem:correctness}
\cref{alg:inc_flow} maintains a $(1-\eps)$-approximate $s$-$t$ max flow at all times with probability $1 - 4/n$ against an output-adaptive adversary.
\end{restatable}
\begin{proof}
In any phase $k$, \cref{lem:singlecorrectness} ensures that an augmenting path exists in the sample at the first time step $i$ when $F < (1-\eps) F_i^*$ with probability $1 - 4/n^2$.
Each time an augmenting path is found, $F$ increases by $1$. Thus, we need to invoke \cref{lem:singlecorrectness} for at most $F^* \le n$ phases. Taking a union bound over all the phases, we get that the approximation guarantee holds throughout the entire insertion sequence with probability $1 - 4/n$ as required.
\end{proof}
\section{Incremental Nagamochi-Ibaraki indices}
\label{sec:inc_ni}

In this section, we show how to incrementally maintain (and sample from) Nagamochi-Ibaraki indices $\ell_e$ for each edge $e$.
The main result of this section is the following:
\NIsample*

We maintain Nagamochi-Ibaraki indices instead of other sampling parameters because of their stability under incremental updates. When an edge $e = (u,v)$ is inserted into a graph $G$, it could change other connectivity parameters (such as \edgeconnectivity, strong-connectivity, or the effective resistance) of \emph{every existing edge} in the graph. However, the \niindex of every other edge remains the same on edge insertion, which is a very desirable property for dynamic algorithms. This happens because while other graph properties are unique for an edge in a given graph, the \niindex depends heavily on the particular forest packing used. It could vary wildly for the same edge, and could range from $1$ up to $n$ depending on the packing.

The incremental algorithm involves maintaining a collection of forests, where each forest is a union-find data structure on the vertices currently in that forest. On an edge insertion, we binary search across the forests to find the first one where its endpoints are disconnected, and insert the edge into that forest (and also adding the endpoints to the forest or initializing a new forest if necessary).

Due to the subtleties of implementing this, we present the data structure in full detail below. In particular, if we na\"ively ``initialized'' a new forest by inserting \emph{all} the vertices into the forest, then this would lead to a running time of $\Omega(n)$ per forest, which could be prohibitive. For the extreme case where $n$ parallel edges are inserted between the same two vertices $u$ and $v$, this necessitates initializing $\Omega(n)$ forests, which leads to $\Omega(n^2)$ total time for the na\"ive implementation, as opposed to adding a vertex to a forest only when necessary (as done below with the $last\_tree$ variable).

\begin{algorithm}[ht]
\DontPrintSemicolon
\caption{Maintaining Incremental Nagamochi-Ibaraki Indices (\textsc{IncNIIndex}) }
\label{alg:inc_ni}
\Fn{\fn{Initialize}}{
    $k \gets 0$ \algcomment{Number of forests maintained}
    $tree \gets []$ \algcomment{Collection of forests}
    $last\_tree[v] \gets 0$ for all $v \in V$ \algcomment{last forest that $v$ belongs to}
}
\Fn{\fn{Insert}($e = (u,v)$)}{
    $i \gets $\fn{FindTree($u, v, \min(last\_tree(u), last\_tree(v)))$)} \algcomment{find correct forest for $e$}
    \textbf{if} {$i > k$} \textbf{then} {$k \mathrel+= 1$, $tree[i] \gets \unionfind.\fn{Initialize}$} \algcomment{add new forest}
    \For{$x \in \{ u,v \}$}{
        \textbf{if} {$i > last\_tree(x)$} \textbf{then} {$last\_tree[x] \mathrel+= 1$, $tree[i].\textsc{Add}(x)$} \algcomment{add $x$ to forest $i$}
    }
    $tree[i].\textsc{Union}(u, v)$ \algcomment{connect $u$ and $v$ in forest $i$}
    \textbf{return} $i$
}
\Fn{\fn{FindTree}($u, v, upper\_bound$)} {
    $L \gets 0$, $R \gets upper\_bound$

    \While{$L \le R$} {
        $M \gets \left\lfloor (L+R) / 2 \right\rfloor $ \algcomment{binary search for forest}
        \textbf{if} \fn{IsConnected}$(u, v, M)$ \textbf{then} $L \gets M+1$ \textbf{else} $R \gets M-1$
    }
    \textbf{return} $L$
}
\Fn{\fn{IsConnected}($u, v, i$)}{
    \textbf{if} $i = 0$ \textbf{then return} \textsc{True}
    \textbf{else return} $tree[i].\fn{Find}(u) = tree[i].\textsc{Find}(v)$
}

\end{algorithm}

\begin{algorithm}[ht]
\DontPrintSemicolon
\caption{Incremental Nagamochi-Ibaraki sampling (\fn{IncNISample})}
\label{alg:inc_ni_sampling}

\Fn{\fn{Initialize}}{
    $X \gets $ \hyperref[alg:inc_ni]{\textcolor{black}{\fn{IncNIIndex}}}.\fn{Initialize}\;
    $Y \gets $ \BST.\fn{Initialize}\;
    $L \gets 0$\;
}

\Fn{\fn{Insert}($e = (u,v)$)}{
    $\ell_e \gets \text{X}.\fn{Insert}(e) $\;
    $Y.\fn{Insert}(e, L)$\;
    $L \gets L + \ell_e^{-1}$\;
}

\Fn{\fn{Sample}}{
    $z \gets$ uniform random number in $[0, L]$\;
    $e \gets Y.\fn{Search}(z)$\;
    \textbf{return} $e$
}
\end{algorithm}

\begin{restatable}{lemma}{NImaintain}
\label{lem:NImaintain}
\cref{alg:inc_ni} maintains the NI forest packing at all times, and runs in total time $O(m \cdot \log m \cdot \alpha(n))$
\end{restatable}
\begin{proof}
We first argue about correctness by induction. It is clearly true before any edge insertions. Suppose it was true until $k$ insertions. Let $e=(u,v)$ be the $(k+1)$-th inserted edge. If $i$ is the index of the tree returned by \textsc{FindTree}, then it satisfies the property that for all $j < i$, vertices $u$ and $v$ are connected in forest $tree[j]$, and that they are not connected in forest $tree[i]$. Thus connecting $u$ and $v$ in forest $tree[i]$ and setting the index of the edge to $i$ maintains correctness after the $(k+1)$-th edge as well.

Next, we argue about the time taken by the algorithm. If there are $m$ edges in the graph currently, then the algorithm maintains $\le m$ forests. Thus the binary search takes $O(\log m)$ iterations. In each iteration, it queries a union-find data structure if two vertices are connected, which takes $O(\alpha(n))$ amortized time. Union of two elements in, insertion of new elements into, and initialization of a union-find data structure takes constant time. Thus the algorithm runs in total time $O(m \cdot \log m \cdot \alpha(n))$.
\end{proof}

We quickly recall an example of a binary search tree insertion sequence before we prove the next lemma. In a binary search tree, if there are three consecutive insertions of (key, value) pairs $(a, 0), (b, 1/3),$ and $(c, 1/2)$, then searching for any $z \in [0, 1/3)$ returns $a$, searching for $z \in [1/3, 1/2)$ returns $b$, and searching for $z \in [1/2, \infty)$ returns $c$.

\begin{restatable}{lemma}{NIsampleonly}
\label{lem:NIsampleonly}
\cref{alg:inc_ni_sampling} maintains a sample from the correct distribution at all times, each insertion takes time $O(\log (m) \alpha(n))$ and each sample takes time $O(\log n)$.
\end{restatable}
\begin{proof}

Let $e_1, \ldots, e_m$ be the sequence of edge insertions performed, and let $\ell_{e_i}$ be their corresponding \niindices obtain from \cref{alg:inc_ni}.
The $\ell_e$ values obtained from \cref{alg:inc_ni} are correct by the guarantees of \cref{lem:NImaintain}.
Define $s_1 = 0$ and $s_i = \sum_{1 \le j < i} \ell_{e_j}^{-1}$ be the sum of inverse \niindices of all edges inserted until $e_{i-1}$.
We will show that at any time $k$, for all $i \in \{1, \ldots, k\}$, searching for any $z \in [s_i, s_{i+1}] \subseteq [0, s_{k+1}]$ returns key $e_i$ in the binary search tree.
This then gives correctness of the algorithm since, after the $k$-th insertion, the edge $e_i$ is sampled with probability $(s_{i+1} - s_i) / L = \ell_{e_i}^{-1} / s_{k+1}$ since $L = s_{k+1}$ after $k$ edge insertions.

The statement is clearly true after the first edge insertion.
Suppose it was true until $k-1$ insertions.
Since the key $e_{k-1}$ had value $s_{k-1}$, searching for any $z \in [s_{k-1}, \infty)$ would have given key $e_{k-1}$.
Let $e_k$ be the $k$-th inserted edge, which is inserted with value $s_{k}$. Thus, for all $z \in [s_k, \infty)$, the search returns key $e_k$, and for $z \in [s_{k-1}, s_k)$, the search returns key $e_{k-1}$. As $s_k - s_{k-1} = \ell_{e_{k-1}}^{-1}$, $s_{k+1} - s_k = \ell_{e_k}^{-1}$, and the other edges $e_j$ for $j < k-1$ are not affected by this insertion, the claim follows.

For time bounds, note that each edge insertion into \cref{alg:inc_ni} takes time $O(\log (m) \alpha(n))$, each insertion into the binary search tree with length $\ell_e^{-1}$ takes time $O(\log n)$, and each sample takes time $O(\log n)$ as well, given that the random number is generated in time $O(\log n)$.
\end{proof}

\cref{lem:NI_sample} now follows from \cref{lem:NImaintain} and \cref{lem:NIsampleonly}.
\section{Balanced sparsification}
\label{sec:balanced_sparsification}

While Cen et al.~\cite{CCPS21cutbalance} obtain a cut sparsification result for balanced directed graphs, their result only works with strong connectivity and cannot be used with NI indices, which we need for our algorithm.
In this section, we show the following result.

\generalsparsification*

This is an extension of the result of \cite{FHHP19sparsification} to balanced directed graphs, with the sampling scheme changed from sampling each edge independently with probability $\approx \polylog (n)/(\lambda_e \eps^2)$ to sampling $\approx O( n \polylog(n)/ \eps^2)$ edges with probability distribution proportional to $\lambda_e^{-1}$. The proof proceeds by reducing the problem to showing concentration for each \connectivityclass where the connectivity parameters are close to each other. Inside each \connectivityclass $F_i$, the cuts are then collected by the size of the underlying undirected cut in $G_i$ and concentration is shown for each collection separately. This involves combining the concentration for each cut in the collection, along with a cut-counting argument similar to the one on \cite{FHHP19sparsification} (which is an extension of Karger's cut-counting). Finally, the concentration of each cut is shown using a standard Chernoff bound.

Concretely, we will show later that the following concentration holds for every \connectivityclass.

\begin{restatable}{lemma}{connectivityclassconc}
\label{lem:connectivityclassconc}
It holds for all $i \in [\Lambda]$ and for all cuts $\vec{C} = \partial_{\vec{G}}(S)$ simultaneously with probability at least $1 - 4/n^2$ that
\[
\bound,
\]
where $C = \partial_G(S)$ is the underlying undirected cut of $\vec{C}$ in $G$.
\end{restatable}

We first use it to prove \cref{thm:generalsparsification}.

\begin{proof}[Proof of \cref{thm:generalsparsification}]
We group the edges by the probability it is picked, and show concentration for each such \connectivityclass separately.
With probability $\ge 1 - 4/n^2$, we have for all cuts $\vec{C} = \partial_{\vec{G}}(S)$ of \balance $\beta^{-1}$,
\begin{align*}
\abs{w(\vec{C}) - u(\vec{C})}
&= \abs{\sum_{i = 0}^\Lambda w(\vec{C} \cap F_i) - \sum_{i=0}^\Lambda u(\vec{C} \cap F_i)} \\
&\le \sum_{i = 0}^\Lambda  \abs{w(\vec{C} \cap F_i) - u(\vec{C} \cap F_i)} \\
&\le \frac{\eps}{2} \cdot \left( \sum_{i=0}^\Lambda \frac{u(C \cap E_i) \cdot 2^{i-1}}{\pi_i \cdot \gamma \cdot (\beta+1)} + \sum_{i=0}^\Lambda u(\vec{C} \cap F_i) \right) \tag*{(by \cref{lem:connectivityclassconc})} \\
&\le \frac{\eps}{2} \cdot \left( \frac{u(C)}{\beta+1} + u(\vec{C}) \right) \\
&= \frac{\eps}{2} \cdot \left( \frac{u(\vec{C}) + u(\cev{C})}{\beta+1} + u(\vec{C}) \right) \\
&\le \eps \cdot u(\vec{C}) \tag*{(since $u(\cev{C}) \le \beta \cdot u(\vec{C})$)}
\end{align*}
where the penultimate inequality follows since
\[
\sum_{i=0}^\Lambda \frac{u(C \cap E_i) \cdot 2^{i-1}}{\pi_i \cdot \gamma} \le u(C)
\] by \gammaoverlap and
\[
\sum_{i=0}^\Lambda u(\vec{C} \cap F_i) = u(\vec{C})
\] since $\mathcal F = \{ F_i \}$ is a partition of $E$.
\end{proof}

To show \cref{lem:connectivityclassconc}, we partition the cuts in a connectivity class based on the size of the underlying undirected cut, and show concentration in each of them separately. Formally, we will show later that
\begin{restatable}{lemma}{sectionconc}
\label{lem:sectionconc}
Let $\mathcal C_{ij}$ be the collection of all cuts $\vec{C} = \partial_{\vec{G}}(S)$ such that for the corresponding undirected cut $C = \partial_G(S)$ in $G$,
\[
\pi_i \cdot 2^j \le u(C \cap E_i) \le \pi_i \cdot 2^{j+1} - 1.
\]
Then for all cuts $\vec{C}$ in $\mathcal C_{ij}$ simultaneously, it holds with probability $1 - 2/n^{4 \cdot 2^j}$ that
\[
\bound.
\]
\end{restatable}

We first use this to prove \cref{lem:connectivityclassconc}.

\connectivityclassconc*

\begin{proof}
We will show that the statement holds for any fixed $i$ with probability at least $1 - 4/n^4$.
Since $\mathcal F = \{ F_i \}$ is a partition of the edges, and since there are at most $n^2$ edges in the graph, a union bound over all $i$ with non-empty $F_i$ gives the lemma.

Fix any $i \in [\Lambda]$. For all cuts $\vec{C}$ with no edge in class $F_i$, the claim holds with probability $1$. In what follows, we only consider cuts that have at least one edge in class $F_i$.

We concentrate on the graph $G_i$, and only consider cuts that contain at least $\pi_i$ edges in $E_i$, since every edge in $F_i$ has connectivity at least $\pi_i$ in $G_i$ by \piconnectivity.
Let $\mathcal C_{ij}$ be the collection of all cuts $\vec{C} = \partial_{\vec{G}}(S)$ such that for the corresponding undirected cut $C = \partial_G(S)$ in $G$,
\[
\pi_i \cdot 2^j \le u(C \cap E_i) \le \pi_i \cdot 2^{j+1} - 1
\]
Then by \cref{lem:sectionconc}, for all $\vec{C} \in \mathcal C_{ij}$, it holds with probability at least $1 - 2/n^{4 \cdot 2^j}$ that
\[
\bound
\]
Thus the probability that there exists a $j$ for which there exists a cut $\vec{C} \in \mathcal C_{ij}$ where the bound does not hold is at most
\[
\frac{2}{n^4} \cdot
\left( \frac{1}{n^{2^0}} + \frac{1}{n^{2^1}} + \ldots \right)
\le \frac{4}{n^4}
\]
as required.
\end{proof}

To show \cref{lem:sectionconc}, we show the bound for a single cut in \cref{lem:singlecutconc} and use this with a directed cut counting argument as in \cref{lem:directed_cut_counting}.

\begin{restatable}{lemma}{singlecutconc}
\label{lem:singlecutconc}
For any single cut $\vec{C}$ in $\mathcal C_{ij}$, it holds with probability $1 - 2/n^{8 \cdot 2^j}$ that
\[
\bound.
\]
\end{restatable}

\begin{restatable}{lemma}{directedcutcounting}
\label{lem:directed_cut_counting}
Let $\vec{G} = (V, \vec{E})$ be a directed (multi-)graph, and $G$ be the underlying undirected graph. Let $k \ge \lambda$ be any real number, where $\lambda$ is the value of the global minimum cut in $G$. Consider all the cuts $\vec{C} = \partial_{\vec{G}}(S)$ in $\vec{G}$ such that $u(C) \le \alpha k$ for the corresponding undirected cut $C = \partial_G(S)$, and let $\mathcal C^{\downarrow k}_\alpha$ be the set of all directed \kprojection of these cuts. Then $|\mathcal C^{\downarrow k}_\alpha| \le 2 \cdot n^{2\alpha}$.
\end{restatable}

We first use them to show \cref{lem:sectionconc}.

\sectionconc*

\begin{proof}
For any single cut $\vec{C}$ in $\mathcal C_{ij}$, \cref{lem:singlecutconc} shows that with probability at least $1 - 2/n^{8\cdot 2^j}$,
\[
\bound
\]
At this point, we would like to take a union bound over all the cuts in $\mathcal C_{ij}$. Ideally, we would like the number of cuts in this set to grow as $n^{c \cdot 2^j}$.
However, the number of directed cuts $\vec{C}$ that arise might be much larger than that. We can circumvent this, since the inequality only depends on $\vec{C} \cap F_i$, and not $\vec{C}$. We thus focus on working with the distinct subsets $\vec{C} \cap F_i$ that arise, and use these to prove our claim.

We first show that the number of \emph{distinct directed cuts $\vec{C} \cap F_i$ that are encountered} is still $O(n^{c \cdot 2^j})$ using \cref{lem:directed_cut_counting}. In particular, we apply \cref{lem:directed_cut_counting} on the graph $G_i = (V, E_i)$, setting $k = \pi_i$ and $\alpha k = \pi_i \cdot 2^{j+1} - 1$.
The total number of distinct subsets $\vec{C} \cap H_k$ that arise from cuts in $\mathcal C_{ij}$ is at most
\[
2 n^{2 \alpha} < 2 n^{4 \cdot 2^j}
\]
where $H_k$ is the subgraph of $G_i$ of all edges with edge-connectivity at least $\pi_i$.
Since each edge in $F_i$ has edge-connectivity at least $\pi_i$ by \piconnectivity, we have that $F_i \subseteq H_k$. Every non-empty subset of the form $\vec{C} \cap F_i$ has a corresponding subset of the form $\vec{C} \cap H_k$ that is counted above, and for every $\vec{C} \cap H_k$ counted above, there is at most one non-empty subset of the form $\vec{C} \cap F_i$, which proves the claim.

To simplify notation, we use the following definition for the next claim.
Call a subset of directed edges $\vec{R} \subseteq \vec{E}$ to be \emph{bad} for a directed cut $\vec{C} \subseteq \vec{E}$ if $\vec{C} \cap F_i = \vec{R}$ and
\[
\abs{w(\vec{R}) - u(\vec{R})} > \frac{\eps}{2} \cdot \left( \frac{u(C \cap E_i) \cdot 2^{i-1}}{\pi_i \cdot \gamma \cdot (\beta+1)} + u(\vec{R}) \right).
\]
Note that if $\vec{R}$ is bad for any particular directed cut $\vec{C}$, then it is also bad for
\[
\vec{D} = \argmin_{\vec{C'}: (\vec{C'} \cap F_i = \vec{R}) \land (\vec{C'} \in \mathcal C_{ij})} u(C' \cap E_i)
\]
since
\[
\abs{w(\vec{R}) - u(\vec{R})}
> \frac{\eps}{2} \cdot \left( \frac{u(C \cap E_i) \cdot 2^{i-1}}{\pi_i \cdot \gamma \cdot (\beta+1)} + u(\vec{R}) \right)
\ge \frac{\eps}{2} \cdot \left( \frac{u(D \cap E_i) \cdot 2^{i-1}}{\pi_i \cdot \gamma \cdot (\beta+1)} + u(\vec{R}) \right),
\]
where the first inequality is because $\vec{R}$ is bad for $\vec{C}$, and the second inequality is by choice of $\vec{D}$.

Thus, the probability we need to bound is
\begin{align*}
&\ \ \ \ \, \Pr [\text{$\exists$ a cut $\vec{C} \in \mathcal C_{ij}$ such that $\vec{C} \cap F_i$ is bad for $\vec{C}$}] \\
&= \Pr [\text{$\exists$ an $\vec{R} \subseteq \vec{E}$ and a cut $\vec{C} \in \mathcal C_{ij}$ such that $\vec{R}$ is bad for $\vec{C}$}] \\
&\le \sum_{\vec{R} \subseteq \vec{E}} \Pr[\text{$\exists$ a cut $\vec{C} \in \mathcal C_{ij}$ such that $\vec{R}$ is bad for $\vec{C}$}] \\
&= \sum_{\vec{R} \subseteq \vec{E}} \Pr[\text{$\vec{R}$ is bad for $\vec{D} = \argmin_{\vec{C'}: (\vec{C'} \cap F_i = \vec{R}) \land (\vec{C'} \in \mathcal C_{ij})} u(C \cap E_i)$}] \\
&\le 2 n^{4 \cdot 2^j} \cdot \frac{2}{n^{8\cdot 2^j}}
\le \frac{4}{n^{4\cdot 2^j}},
\end{align*}
where the penultimate inequality holds because the number of distinct $\vec{R}$s that can be of the form $\vec{C} \cap F_i$ is at most $2 n^{4 \cdot 2^{j}}$, and for each $\vec{R}$, the probability that $\vec{R}$ is bad for $\vec{D}$ is at most $2/n^{8 \cdot 2^j}$ by \cref{lem:singlecutconc}.
\end{proof}

We will prove a Chernoff bound of the following form in Appendix~\ref{sec:chernoff}.

\begin{restatable}{lemma}{chernoffbound}
\label{lem:chernoffbound}
Let $X_1, \ldots, X_n$ be $n$ independent random variables that take values in the range $[0,M]$ for some $M > 0$.
Then, for any $\eps \in (0, 1)$ and $N \ge E[\sum_i X_i]$, we have
\[
\Pr \left[ \abs{\sum_{i=1}^n X_i - E\left[ \sum_{i=1}^n X_i \right]} > \eps N \right]
\le 2 \exp(- 0.38 \eps^2 N / M).
\]
\end{restatable}

Using this, we prove \cref{lem:singlecutconc}.

\singlecutconc*
\begin{proof}
For each round $j \in \{ 1,2, \ldots \rho\}$, let $X_j$ be the following random variable
\[
X_j =
\begin{cases}
    w_e = (\rho \cdot p_e)^{-1} & \text{if an edge $e \in \vec{C} \cap F_i$ is chosen in round $j$} \\
    0 & \text{if an edge $e \not\in \vec{C} \cap F_i$ is chosen in round $j$}
\end{cases}
\]
Defining $X = \sum_{j=1}^\rho X_j$, we have that $X = w(\vec{C} \cap F_i)$ and $E[X] = u(\vec{C} \cap F_i)$.
Since $\lambda_e < 2^{i+1}$ for any edge $e \in F_i$, we have that
\[
p_e
\ge \frac{\lambda_e^{-1}}{\sum_{e \in E} \lambda_e^{-1}}
\ge \frac{1}{2^{i+1} \cdot \sum_{e \in E} \lambda_e^{-1}}.
\]
This gives an upper bound on the weight as
\[
\max_j |X_j|
= \max_{e \in \vec{C} \cap F_i} w_e \\
= \max_{e \in \vec{C} \cap F_i} \frac{1}{\rho \cdot p_e}
\le \frac{2^{i+1} \cdot \sum_{e \in E} \lambda_e^{-1}}{\rho}.
\]
Recall that in the algorithm, we set
\[
\rho = \generalsamplesize.
\]
Substituting this into the upper bound on the weight, we apply \cref{lem:chernoffbound} to the random variables $X_i$ using the following explicit values of $M$ and $N$.
\[
M = \left( \frac{128 \gamma (\beta+1) \ln n}{0.38 \cdot 2^{i+1} \eps^2} \right)^{-1}
\qquad \text{and} \qquad
N = u(\vec{C} \cap F_i) + \frac{u(C \cap E_i) \cdot 2^{i-1}}{\pi_i \cdot \gamma \cdot (\beta+1)}.
\]
It trivially holds that
\[
N = u(\vec{C} \cap F_i) + \frac{u(C \cap E_i) \cdot 2^{i-1}}{\pi_i \cdot \gamma \cdot (\beta+1)}
\ge u(\vec{C} \cap F_i) = E[X].
\]
Thus
\begin{align*}
&\Pr \left[ \abs{w(\vec{C} \cap F_i) - u(\vec{C} \cap F_i)} > \frac{\eps}{2} \cdot \left( \frac{u(C \cap E_i) \cdot 2^{i-1}}{\pi_i \cdot \gamma \cdot (\beta+1)} + u(\vec{C} \cap F_i) \right) \right] \\
&\le 2 \exp\left(- 0.38 \cdot \left(\frac{\eps}{2}\right)^2 \cdot
\left( \frac{128 \gamma (\beta+1) \ln n}{0.38 \cdot 2^{i+1} \eps^2} \right) \cdot \left( \frac{u(C \cap E_i) \cdot 2^{i-1}}{\pi_i \cdot \gamma \cdot (\beta+1)} + u(\vec{C} \cap F_i) \right)
\right) \\
&\le 2 \exp\left(- 0.38 \cdot \left(\frac{\eps}{2}\right)^2 \cdot
\left( \frac{128 \gamma (\beta+1) \ln n}{0.38 \cdot 2^{i+1} \eps^2} \right) \cdot \left( \frac{u(C \cap E_i) \cdot 2^{i-1}}{\pi_i \cdot \gamma \cdot (\beta+1)}\right)
\right) \\
&\le 2 \exp\left( -8 \cdot \frac{u(C \cap E_i) \cdot \ln n}{\pi_i} \right) \\
&\le 2 \exp\left( -8 \cdot 2^j \cdot \ln n \right) \tag*{(since $\vec{C} \in \mathcal C_{ij} \implies u(C \cap E_i) \ge \pi_i \cdot 2^j$)}
\end{align*}
as required.
\end{proof}

Finally, we use the undirected projection result of \cite{FHHP19sparsification} to prove \cref{lem:directed_cut_counting}.

\directedcutcounting*
\begin{proof}
We first apply \cref{lem:fhhpprojection} on the undirected graph $G$.
Consider the set of all undirected cuts $C = \delta(S, \bar{S})$ in $G$ such that $u(C) \le \alpha k$.
\cref{lem:fhhpprojection} gives that there are at most $n^{2\alpha}$ undirected $k$-projections of such cuts, i.e., $|\mathcal K^{\downarrow k}_{\alpha}| \le n^{2\alpha}$.
We will show that $|\mathcal C^{\downarrow k}_{\alpha}| \le 2 \cdot |\mathcal K^{\downarrow k}_{\alpha}|$, which gives the lemma.

We define a function from the first set to the second set where each image has a pre-image set of size at most two.
Recall that $H_k$ is the set of edges in $G$ of edge-connectivity at least $k$.
We map each directed projection $\vec{C} \cap H_k$ in $\mathcal C^{\downarrow k}_\alpha$ to the corresponding undirected projection $C \cap H_k$. The latter belongs to $\mathcal K^{\downarrow k}_\alpha$ since $u(C) \le \alpha k$ by definition. For each $C \cap H_k \in \mathcal K^{\downarrow k}_\alpha$, there are at most two directed projections that map to it since each undirected cut can be directed as either $(S, \bar{S})$ or $(\bar{S}, S)$.
\end{proof}
\section*{Acknowledgements}
\label{sec:ack}

MH: This project has received funding from the European Research Council (ERC) under the European Union's Horizon 2020 research and innovation programme (MoDynStruct, No. 101019564)  \includegraphics[width=0.9cm]{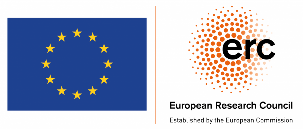} and the Austrian Science Fund (FWF) grant  \href{https://www.doi.org/10.55776/Z422}{DOI 10.55776/Z422}, grant  \href{https://www.doi.org/10.55776/I5982}{DOI 10.55776/I5982}, and grant  \href{https://www.doi.org/10.55776/P33775}{DOI 10.55776/P33775} with additional funding from the netidee SCIENCE Stiftung, 2020–2024.
For open access purposes, the author has applied a CC BY public copyright license to any author-accepted manuscript version arising from this submission. Views and opinions expressed are however those of the author(s) only and do not necessarily reflect those of the European Union or the European Research Council Executive Agency. Neither the European Union nor the granting authority can be held responsible for them.
\bibliographystyle{gamma}
\def\bibfont{\small}
\bibliography{Ref}

\appendix
\section{Chernoff Bounds}
\label{sec:chernoff}

We use the following Chernoff bound from~\cite{FHHP19sparsification} to obtain \cref{lem:chernoffbound}.

\begin{restatable}[\cite{FHHP19sparsification}]{lemma}{fhhpchernoffbound}
\label{lem:fhhpchernoffbound}
Let $Y_1, \ldots, Y_n$ be $n$ independent random variables each taking values in $[0,1]$. Let $\mu = E[\sum_i Y_i]$. Then for all $\delta > 0$, we have
\[
\Pr \left[ \abs{\sum_{i =1}^n Y_i - \mu} > \delta \mu \right]
\le 2 \exp \left( - 0.38 \min \{ 1, \delta \} \delta \mu \right).
\]
\end{restatable}

Recall the bound we want to prove.

\chernoffbound*
\begin{proof}
Define $Y_i = M^{-1} \cdot X_i$. Then $Y_i \in [0,1]$ and $\mu = E[\sum_i Y_i] = M^{-1} \cdot E[\sum_i X_i]$. Applying \cref{lem:fhhpchernoffbound} for $\delta = \frac{\eps N}{E[\sum_i X_i]}$, we get that
\[
\Pr \left[ \abs{\sum_{i = 1}^n X_i - E\left[ \sum_{i = 1}^n X_i \right] } > \eps N \right]
\le 2 \exp \left( - 0.38 \min \{ 1, \delta \} \delta \mu \right).
\]
If $\delta \le 1$, then using the fact that $N \ge E[\sum_i X_i]$, the bound becomes
\[
2 \exp \left( - 0.38 \cdot \left( \frac{\eps N}{E[\sum_i X_i]} \right)^2 \cdot M^{-1} \cdot E\left[ \sum_i X_i \right] \right)
\le 2 \exp \left( - 0.38 \eps^2 N / M \right),
\]
and if $\delta > 1$, then the bound becomes
\[
2 \exp \left( - 0.38 \cdot \left( \frac{\eps N}{E[\sum_i X_i]} \right) \cdot M^{-1} \cdot E\left[ \sum_i X_i \right] \right)
\le 2 \exp \left( - 0.38 \eps N / M \right),
\]
where the final bound follows since $\eps^2 < \eps$
for $\eps \in [0,1]$.
\end{proof}

\end{document}